\documentclass[11pt,letterpaper]{article}
\usepackage[margin=1in]{geometry}
\usepackage{amsmath, amssymb,graphicx}
\usepackage{amsthm}
\usepackage{color}
\usepackage{verbatim, xspace}
\usepackage{algorithm}
\usepackage{algpseudocode}
\usepackage{thm-restate}
\usepackage{enumerate} 
\usepackage{subcaption}
\usepackage[hidelinks]{hyperref}

\graphicspath{ {figures/} }
\newtheorem{theorem}{Theorem}[section]
\newtheorem{definition}[theorem]{Definition}
\newtheorem{lemma}[theorem]{Lemma}
\newtheorem{corollary}[theorem]{Corollary}

\newtheorem{proposition}[theorem]{Proposition}

\newtheorem{observation}[theorem]{Observation} 

\newcommand{\ar}{\ensuremath{\text{\textsc{ar}}}\xspace}

\newcommand{\eps}{\varepsilon}
\newcommand{\Oh}{\mathcal{O}}
\newcommand{\zz}{\mathbb{Z}}
\newcommand{\nn}{\mathbb{N}}
\newcommand{\cs}{\textsc{CS}}
\newcommand{\glit}{\Gamma^{LIT}}
\newcommand{\gneg}{\Gamma^{NEG}}
\newcommand{\gtf}{\Gamma^{T,F}}
\newcommand{\NAE}{\text{NAE}}
\newcommand{\OR}{\text{OR}}
\newcommand{\AND}{\text{AND}}
\newcommand{\XOR}{\text{XOR}}
\newcommand{\EX}{\text{EX}}
\newcommand{\maxcsp}{\textsc{Max CSP}\xspace}
\newcommand{\maxcspg}[1]{\textsc{Max CSP($#1$)}\xspace}
\newcommand{\cspnn}[1]{\textsc{Max CSP}$(\Gamma_{#1}, \nn)$}
\newcommand{\cspzz}[1]{\textsc{Max CSP}$(\Gamma_{#1}, \zz)$}
\newcommand{\cspnnc}[1]{\textsc{Max CSP}$(\Gamma_{#1}, \nn, c)$}
\newcommand{\cspzzc}[1]{\textsc{Max CSP}$(\Gamma_{#1}, \zz, c)$}
\newcommand{\cspnnpar}[1]{\textsc{Max CSP}$(#1, \nn, c)$}

\newcommand{\containment}[0]{\ensuremath{\text{NP} \subseteq \text{co-NP} / \text{poly}}\xspace}
\newcommand{\ncontainment}[0]{\ensuremath{\text{NP} \not \subseteq \text{co-NP} / \text{poly}}\xspace}

\usepackage{color}
\ifdefined\DEBUG
  \newcommand{\mic}[1]{\textcolor{red}{#1}}
  \newcommand{\bmp}[1]{{\color{blue}{#1}}}
  \newcommand{\micr}[1]{\marginpar{\small \textcolor{red}{$\bullet$ #1}}}
  \newcommand{\bmpr}[1]{\marginpar{\small \textcolor{blue}{$\bullet$ #1}}}
\else
  \newcommand{\mic}[1]{#1}
  \newcommand{\bmp}[1]{#1}
  \newcommand{\micr}[1]{}
  \newcommand{\bmpr}[1]{}
\fi

\title{Optimal polynomial-time compression for Boolean Max CSP\footnote{This project has received funding from the~European Research Council (ERC) under the~European Union's Horizon 2020 research and innovation programme (grant agreement No 803421, ReduceSearch).}}
\author{Bart M.P. Jansen\footnote{Address: \texttt{b.m.p.jansen@tue.nl}}  \\ Eindhoven University of~Technology
\and Micha{\l} W{\l}odarczyk\footnote{Address: \texttt{m.wlodarczyk@tue.nl}} \\ Eindhoven University of~Technology}
\date{}

\begin{document}
\maketitle

\begin{abstract}
In the~Boolean maximum constraint satisfaction problem -- \textsc{Max CSP$(\Gamma)$} -- one is given a~collection of~weighted applications of~constraints from a~finite constraint language~$\Gamma$, over a~common \bmp{set} of~variables, and the~goal is to assign Boolean values to the~variables so that the~total weight of~satisfied constraints is maximized.
There exists an elegant dichotomy theorem providing a~criterion on $\Gamma$ for the~problem to be polynomial-time solvable and stating that otherwise it becomes NP-hard.
We study the~NP-hard \bmp{cases} through the~lens of~kernelization and provide a~complete characterization of~\textsc{Max CSP$(\Gamma)$} with respect to the~optimal compression size.
Namely, we~prove that \textsc{Max CSP$(\Gamma)$} parameterized by the~number of~variables~$n$ is either polynomial-time solvable, or there exists an integer $d \ge 2$ depending on $\Gamma$, such that
\begin{enumerate}
    \item \bmp{An instance of} \textsc{Max CSP$(\Gamma)$} can be compressed into an equivalent instance with $\Oh(n^d\log n)$ bits in polynomial time,
    \item \textsc{Max CSP$(\Gamma)$} does not admit such a~compression to $\Oh(n^{d-\eps})$ bits unless $\text{NP} \subseteq \text{co-NP} / \text{poly}$.
\end{enumerate}

Our reductions are based on interpreting constraints as multilinear polynomials combined with the~framework of~constraint implementations.
As another application of~our reductions, 
we~reveal tight connections between optimal running times for solving \textsc{Max CSP$(\Gamma)$}.
More precisely,
we~show that obtaining a~running time of~the~form $\Oh(2^{(1-\eps)n})$ for particular classes of~\bmp{\textsc{Max CSP}s} is as hard as breaching this barrier for \textsc{Max $d$-SAT} for some $d$.\micr{I decided to drop the~notation $\Oh^*$ because the~polynomial term can be dominated by adjusting $\eps$.}
\end{abstract}

\includegraphics[height=1cm]{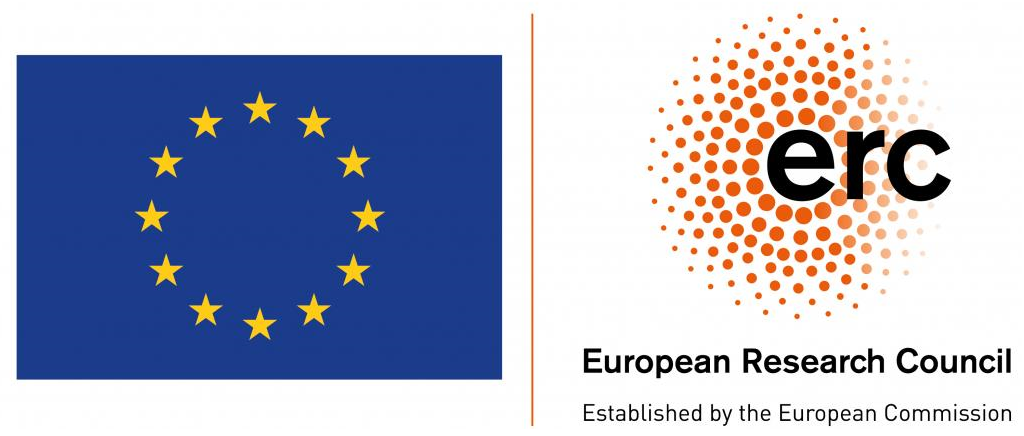}

\section{Introduction}
\paragraph*{Background and motivation}
The framework of~constraint satisfaction problems (CSPs) allows the~computational complexity of~a~large class of~problems to be studied through a~common lens~\cite{CreignouKS01}. A~typical instance of~such a~problem asks whether it is possible to assign each of~the~variables~$x_1, \ldots, x_n$ a~value from a~finite domain~$D$, such that a~given list of~constraint applications is satisfied. A~constraint is applied to a~fixed number of~variables, and indicates which combinations of~values are legal. In the~\maxcsp problem, the~goal is to maximize the~number of~satisfied constraints. See Section~\ref{sec:preliminaries} for formal definitions.

The investigation of~CSPs has led to deep theorems characterizing the~complexity of~a~CSP based on the~type of~constraints allowed in the~instance~\cite{BulatovGKK09,KrokhinZ17}. For example, the~long-awaited CSP dichotomy theorem~\cite{Bulatov17,Zhuk17} provides a~criterion separating the~NP-complete from the~polynomial-time solvable CSPs; the~work of~Khanna, Sudan, Trevisan, and Williamson characterizes how well the~maximization version of~a~Boolean CSP can be approximated~\cite{KhannaSTW00} (see~\cite{DeinekoJKK08, JonssonK07} for larger domains; see~\cite{DalmauKM18,MakarychevM17} for optimal approximation factors); and Cai and Chen~\cite{CaiC17} present a~dichotomy that separates CSPs for which the~number of~complex-weighted solutions can be counted in polynomial time, from those where the~problem is \#P-hard.


In this work we~analyze the~complexity of~constraint satisfaction in an algorithmic regime that is currently far from understood: polynomial-time compression and kernelization~\cite{FominLSZ19}. Here, the~goal is to analyze how much (in terms of~the~number of~variables~$n$) an instance can be compressed by a~polynomial-time algorithm without changing the~answer, and to understand how the~compressibility depends on the~type of~available constraints. A~\emph{compression} is a~polynomial-time algorithm that reduces instances of~one problem to equivalent, small instances of~a~potentially different problem; a~\emph{kernelization} compresses to an instance of~the~same problem \bmp{(see Section~\ref{sec:prelims:fpt}).} A~kernelization of~small size allows an instance to be stored, manipulated, and solved more efficiently. It is therefore of~interest to find the~smallest possible kernelizations. Since every kernelization yields a~compression, one can prove lower bounds on the~size of~kernelizations by establishing lower bounds on compressions.

In recent years, there have been a~number of~advances in the~understanding of~compressibility of~CSPs~\cite{ChenJP20,DellM14,JansenP19,LagerkvistW17}. A~foundational result by Dell and van Melkebeek~\cite{DellM14} states that for~$d \geq 3$, \textsc{CNF-SAT} with clauses of~size at most~$d$ (\textsc{$d$-CNF-SAT}) parameterized by the~number of~variables~$n$ admits no \bmp{(polynomial-time)} compression of~size~$\Oh(n^{d-\varepsilon})$ for any~$\varepsilon > 0$, unless \containment. As an instance of~\textsc{$d$-CNF-SAT} can trivially be compressed to~$\Oh(n^d)$ bits via~a~bitstring that encodes for each of~the~$\Oh(n^d)$ possible clauses whether or not it is present in the~instance, the~\textsc{$d$-CNF-SAT} problem does not admit any non-trivial compression. The situation is different for the~related problem \textsc{$d$-Not-All-Equal SAT} (\textsc{$d$-NAE-SAT}), which is the~variant where a~clause is satisfied when its literals do not all evaluate to the~same value. Jansen and Pieterse showed~\cite{JansenP17,JansenP19} that for~$d \geq 3$, the~\textsc{$d$-NAE-SAT} problem has a~compression of~size~$\Oh(n^{d-1} \log n)$, but not of~size~$\Oh(n^{d-1-\varepsilon})$ unless \containment. This example shows that the~type of~constraints affects the~compressibility of~a~CSP.

The notion of~a~\emph{constraint language} is used to rigorously analyze how the~complexity of~a~CSP depends on the~type of~constraints. In this work, we~will only consider CSPs over the~Boolean domain: we~work exclusively with Boolean constraints and constraint languages. \mic{a~constraint is therefore a~function of~the~form~$f \colon \{0,1\}^k \rightarrow \{0,1\}$, where~$k \geq 1$ is the~\emph{arity} of~the~constraint, also denoted as~$\ar(f)$. A~constraint language~$\Gamma$ is a~\emph{finite} set of~constraints. The input of~the~corresponding decision problem, denoted CSP($\Gamma)$\bmp{,} consists of~a~set of~constraint applications of~the~form $f(x_{j_1}, \dots, x_{j_{\ar(f)}}) = 1$ over $n$ common variables, where $f$ is some constraint from $\Gamma$.
The question is whether there is an assignment~$\{x_1,\ldots,x_n\} \to \{0,1\}$
satisfying all the~constraint applications.}

In this terminology, Chen, Jansen, and Pieterse~\cite{ChenJP20} characterized for all (Boolean) constraint languages~$\Gamma$ consisting of~constraints of~arity at most three, what the~optimal compression size is for CSP($\Gamma$). Lagerkvist and Wahlstr\"om~\cite{LagerkvistW17} gave universal-algebraic conditions on~$\Gamma$ which ensure that CSP($\Gamma$) has a~compression of~size~$\Oh(n \log n)$, and a~characterization is known of~the~constraint languages~$\Gamma_{sym}$ consisting entirely of~\emph{symmetric} functions for which CSP($\Gamma_{sym})$ has a~compression of~near-linear size~\cite[\S 5]{ChenJP20}. Hence there is some understanding of~the~optimal compressibility of~CSP($\Gamma$). 

However, when we~move from the~question of~whether \emph{all} constraints can be satisfied to the~task of~maximizing the~number of~satisfied constraints (\maxcsp),  the~situation is much less understood. To the~best of~our knowledge, no non-trivial compressions are known for any \maxcspg{\Gamma}, and no compression lower bounds are known for \maxcspg{\Gamma} other than those already implied from CSP($\Gamma$). In this paper, we~therefore analyze the~compressibility of~\mbox{\maxcspg{\Gamma}}.

Before presenting our results, we~briefly summarize the~main algorithmic approach for compressing CSP($\Gamma$) and illustrate why it fails completely for \maxcsp. Consider for example 3-NAE-SAT. The number of~constraint applications in an $n$-variable instance of~this problem can be reduced to~$\Oh(n^2)$ without changing the~solution space, which allows it to be encoded in~$\Oh(n^2 \log n)$ bits. The \emph{sparsification} to~$\Oh(n^2)$ constraint applications is achieved by a~linear-algebraic approach. Note that a~not-all-equal constraint on variables~$(x,y,z) \in \{0,1\}^3$ is satisfied if and only if~$x + y + z - xy - xz - yz - 1 = 0$. Observe that if~$p_1(x_1, \ldots, x_n) = 0, \ldots, p_m(x_1, \ldots, x_n) = 0$ are polynomial equalities which are satisfied by an assignment to~$x_1, \ldots, x_n$, then also~$\sum _{i=1}^m \alpha_i \cdot p_i(x_1, \ldots, x_n) = 0$ holds for any linear combination as determined by~$\alpha_1, \ldots, \alpha_m$. To sparsify a~3-NAE-SAT instance with this insight, proceed as follows. Transform each constraint~$c_i$ into an equality~$p_i(x_1, \ldots, x_n) = 0$ for a~degree-2 polynomial~$p_i$, substituting~$1-v$ for negated variables~$\neg v$ in the~constraint. This yields a~system of~equations of~degree-2 polynomials in~$n$ variables, which \bmp{have}~$\Oh(n^2)$ distinct monomials. The rank of~a~corresponding vector space is therefore~$\Oh(n^2)$, which yields a~basis of~$\Oh(n^2)$ equalities such that all others can be expressed as their linear combinations. All constraints not corresponding to an element of~this basis can be safely omitted from an instance of~3-NAE-SAT, since they will be automatically satisfied by any assignment that satisfies \emph{all} basis constraints. This yields the~claimed sparsification of~$\Oh(n^2)$ constraints. Note, however, that this approach fails completely for the~variant \textsc{Max 3-NAE-SAT}: if an assignment \emph{does not} satisfy all constraints of~the~basis, this does not give any satisfaction guarantees on the~linearly-dependent constraints. Hence the~sparsification approach for CSP($\Gamma$) is not applicable for \maxcspg{\Gamma}, and a~priori it is not clear whether any \maxcspg{\Gamma} problem admits a~non-trivial polynomial-time compression.

\paragraph*{Our results}
We provide a~new route to compression for \maxcspg{\Gamma}, and prove that the~resulting compressions are \emph{essentially optimal} for all constraint languages~$\Gamma$, assuming \ncontainment. Our results characterize the~optimal compressibility of~all Boolean {\maxcsp}s in terms of~degrees of~characteristic polynomials, and uncover a~wide range of~\maxcspg{\Gamma} problems that admit a~non-trivial compression. For a~Boolean function~$f \colon \{0,1\}^k \to \{0,1\}$, its \emph{characteristic polynomial} is the~\emph{unique} $k$-variate multilinear polynomial~$P_f(x)$ over~$\mathbb{R}$ that agrees with~$f$ on all~$x \in \{0,1\}^k$. The fact that this representation is unique is well-known (cf.~\cite{NisanS94}). For a~constraint language~$\Gamma$, define~$\deg(\Gamma) = \max _{f \in \Gamma} \deg(P_{f})$. We prove that~$\deg(\Gamma)$ characterizes the~compressibility of~\maxcspg{\Gamma}. 

To state our results precisely, we~have to address a~feature of~the~problem that is particular to the~maximization variant: repetitions of~constraint applications. While such repetitions are irrelevant in the~CSP setting when all constraint applications have to be satisfied, they become relevant when maximizing the~number of~satisfied constraint applications. The standard approach in the~\maxcsp literature is therefore to give each constraint application a~positive integer weight value~\cite{CreignouKS01,KhannaSTW00}.
The decision problem \maxcspg{\Gamma} then takes as input a~system of~$\Gamma$-constraint applications with weights from~$\mathbb{N}$, and a~threshold value~$t$, and asks whether there is an assignment such that the~weight of~the~satisfied constraint applications is at least~$t$.

Let~$\Gamma$ be a~(finite, Boolean) constraint language. Our main positive result is the~following.

\begin{theorem} \label{thm:upperbound:informal}
\maxcspg{\Gamma} parameterized by the~number of~variables~$n$, with positive integer weights bounded by~$n^{\Oh(1)}$, admits a~compression of~size~$\Oh\big(n^{\deg(\Gamma)} \log n\big)$.
\end{theorem}

In fact, we~are even able to reduce any instance of~\maxcspg{\Gamma} to an equivalent instance of~the~\emph{same} problem, having~$\Oh\big(n^{\deg(\Gamma)}\big)$ weighted constraint applications. We prove matching lower bounds whenever \maxcspg{\Gamma} is NP-complete. It is known~\cite{Creignou95,CreignouKS01,KhannaSTW00} that for inputs with positive integer weights, \maxcspg{\Gamma} is polynomial-time solvable if~$\Gamma$ is 0-valid, 1-valid, or~2-monotone (see Section~\ref{sec:constraint:types}), and NP-complete otherwise.

\begin{theorem} \label{thm:lowerbound:informal}
If~$\Gamma$ is not 0-valid, 1-valid, or 2-monotone, then assuming \ncontainment, \maxcspg{\Gamma} parameterized by the~number of~variables~$n$, with positive integer weights bounded by~$n^{\Oh(1)}$, does not admit a~compression of~size~$\Oh\big(n^{\deg(\Gamma)-\eps}\big)$ for any~$\eps > 0$.
\end{theorem}

Our results uncover an interesting contrast in compressibility between decision CSPs and maximization CSPs. While both involve the~analysis of~the~degrees of~polynomials, the~type of~polynomials which is used differs, leading to differences in compressibility. The linear-algebraic approach to sparsify CSP($\Gamma$) yields a~compression of~size~$\Oh(n^d \log n)$ when for each constraint~$f \in \Gamma$, for each assignment~$u \in \{0,1\}^{\ar(f)}$ for which~$f(u) = 0$, there exists a~ring~$R$ and a~degree-$d$ polynomial over~$R$ for which~$P(u) \neq 0$ and~$P(x) = 0$ for all~$x \in \{0,1\}^{\ar(f)}$ with~$f(x) = 1$. Different rings can be used for different constraints, and all that matters is that the~polynomial for~$u$ is nonzero on~$u$ but zero on all satisfying assignments. In contrast, for \maxcsp we~can only use polynomials over~$\mathbb{R}$ (or, equivalently, $\mathbb{Q}$), and must ensure that the~value of~the~polynomial coincides with the~value of~the~constraint on all Boolean assignments. This means that a~higher-degree polynomial may be needed, which also translates into worse  compressibility. For example, while \textsc{$d$-NAE-SAT} has a~compression of~size~$\Oh(n^{d-1} \log n)$ for \emph{all}~$d \geq 3$, the~corresponding \textsc{Max $d$-NAE-SAT} problem with weights of~absolute value~$n^{\Oh(1)}$ has a~compression of~size~$\Oh(n^{d-1} \log n)$ for \emph{odd~$d \geq 3$}, but no compression of~size~$\Oh(n^{d-\eps})$ for \emph{even $d$}.
{Another example is \textsc{$d$-Exact SAT}, where we~require exactly one literal in each clause to be true.
Whereas \textsc{$d$-Exact SAT} admits a~compression of~size $\Oh(n \log n)$ for every fixed $d$~\cite{ChenJP20}, we~show that
\textsc{Max $d$-Exact SAT} cannot be compressed to $\Oh(n^{d-\eps})$~bits.}

\paragraph*{Techniques}
On a~high level, our results are obtained by combining two ingredients: (1) a~characterization of~the~complexity of~a~constraint language as~$\deg(\Gamma)$, via~the~degree of~the~characteristic polynomials, and (2) reductions between different problems \maxcspg{\Gamma} and \maxcspg{\Gamma'} by implementing constraints of~one language by combinations of~constraints from the~other. While both ingredients have been used in isolation~\cite{Creignou95,CreignouKS01,KhannaSTW00,LincolnWW18,Williams07}, their combination is novel and is the~key to understanding compressibility. To comprehend how characteristic polynomials help to compress an instance of~\maxcspg{\Gamma}, observe that since the~characteristic polynomial gives~$1$ when a~constraint is satisfied and~$0$ otherwise, the~total value of~satisfied constraint applications can be written as a~weighted sum of~applications of~characteristic polynomials. If~$\deg(\Gamma) = k$, then this weighted sum contains~$\Oh(n^k)$ distinct monomials. An instance can therefore be compressed by expanding this weighted sum, and storing the~coefficient of~each monomial. If all weights in the~input instance are bounded by~$n^{\Oh(1)}$, each coefficient will have value~$n^{\Oh(1)}$ and can therefore be encoded in~$\Oh(\log n)$ bits.

Our lower bounds are obtained by parameterized reductions between {\maxcsp}s in which the~number of~variables does not grow significantly. By a~careful analysis of~the~terms of~the~characteristic polynomial, we~effectively show that if~$\deg(\Gamma) = \deg(\Gamma')$, then constraint applications from~$\Gamma$ can effectively be simulated by combinations of~constraints from~$\Gamma'$. Here, we~use the~framework of~\emph{implementations} from an earlier work~\cite{KhannaSTW00}. Since the~characteristic polynomial of~$d$-CNF clauses has degree~$d$, this yields a~reduction from \textsc{$d$-CNF-SAT} to \maxcspg{\Gamma} for~$\deg(\Gamma) = d$ that preserves the~asymptotic size of~the~variable set, therefore transferring the~cited lower bound for \textsc{$d$-CNF-SAT}~\cite{DellM14} to \maxcspg{\Gamma}. The same reduction is also used to turn the~compression sketched above into a~kernelization, which outputs an instance of~the~original problem.

\paragraph*{Consequences for exponential-time algorithms}
The framework we~develop for parameterized reductions among {\maxcsp}s also has consequences for exponential-time algorithms, which we~believe to be of~independent interest. The \emph{\textsc{Max 3-SAT Hypothesis}}~\cite{LincolnWW18} states that \textsc{Max 3-CNF-SAT} with~$n$ variables cannot be solved in time~$\Oh(2^{(1-\eps)n})$ for any~$\eps > 0$ (cf.~\cite{AlmanW20, BringmannFK19}). Our reductions imply that this hypothesis is \emph{equivalent} to the~version where \textsc{Max 3-CNF-SAT} is replaced by \maxcspg{\Gamma} for any constraint language~$\Gamma$ with~$\deg(\Gamma) = 3$ in which negated literals can be expressed (formal details follow later). In particular, the~\textsc{Max 3-SAT} hypothesis is equivalent to the~statement that \textsc{Max E3-Lin} cannot be solved in time~$\Oh(2^{(1-\eps)n})$.
{What is more, for any $k \ge 2$,
our reductions uncover an equivalence class of~NP-hard problems whose optimal exponential-time running times coincide with the~one for \textsc{Max $k$-SAT}.}

\paragraph*{Related work}
Representations of~Boolean functions as polynomials have been studied frequently in the~literature~\cite{Barrington1994,Beigel93,LincolnWW18, MinP69,NisanS94,TardosB98, GathenR97} {revealing, e.g., a~relation between the~degree of~the~representation and the~decision tree complexity~\cite{NisanS94}.} 
Algorithms for CSPs via~their characteristic polynomials were first given by Williams~\cite{Williams07}. He used the~split-and-list technique to give accelerated exponential-time algorithms for \textsc{Max 2-SAT} and \maxcspg{\Gamma} for~$\deg(\Gamma) = 2$. In~recent work, Lincoln, Williams, and Vassilevska~Williams~\cite{LincolnWW18} give an exponential-time split-and-list reduction from \maxcspg{\Gamma} for~$\deg(\Gamma) = k$ to the~problem of~detecting an $\ell$-hyperclique in a~$k$-uniform hypergraph, for~$\ell > k$, in support of~the~$(k,\ell)$-hypothesis which states that detecting such a~hyperclique in an $n$-vertex input requires time~$n^{\ell-o(1)}$ on a~Word-RAM with~$\Oh(\log n)$-bit words. If~this hypothesis fails for some~$k$ and~$\ell$, their reduction implies that each \maxcspg{\Gamma} problem with~$\deg(\Gamma) = k$ can be solved in time~$\Oh(2^{(1-\eps)n})$ for some~$\eps > 0$. As their reductions run in exponential time, they are very different from ours.



\paragraph*{Organization}
We begin with Section~\ref{sec:preliminaries} containing all the~necessary definitions about CSPs and a~summary of~the~important properties that are already known, including the~implementation framework.
In Section~\ref{sec:polynomials} we~explain the~idea~of~representing constraints by polynomials and provide an algebraic background for our reductions.
It is followed by Section~\ref{sec:reductions}, where the~notion of~reduction between constraint systems is formalized, and the~main reductions are presented.
It serves as a~toolbox for proving the~main results for compression (Section~\ref{sec:compression}) and exponential-time algorithms (Section~\ref{sec:exponential}).

\section{Preliminaries} \label{sec:preliminaries}

For a~set~$S$ and integer~$d \geq 0$, let $\binom{S}{d}$ be the~collection of~all unordered size-$d$ subsets of~$S$. We use~$[n]$ as a~shorthand for~$\{1, \ldots, n\}$. 
A~$k$-ary constraint is a~function $f 
\colon \{0,1\}^k \rightarrow \{0,1\}$.
We refer to $k$ as the~arity of~$f$, denoted $\ar(f)$.
We always assume that the~domain is Boolean. \micr{I removed the~assumption $\ar(f) > 0$ as we~allow trivial constraints in $\gtf$.}
A~constraint $f$ is satisfied by an input $s \in \{0,1\}^k$ if $f(s) = 1$.
A~constraint language (sometimes called constraint family) $\Gamma$ is a~finite collection of~constraints $\{f_1, f_2, 
\dots, f_\ell\}$, potentially with different arities.
\mic{a~\emph{constraint application}, of~a~$k$-ary constraint $f$ to a~set of~$n$ Boolean variables, is a~triple $\langle f, (i_1, i_2, \dots i_k), w \rangle$, where
the~indices $i_j \in [n]$ select $k$ of~the~$n$ Boolean variables to whom the~constraint is applied, and $w$ is a~weight, described formally below.}
The variables can repeat in a~single application.

\begin{definition}
A~constraint system is a~pair $CS(\Gamma, \mathbb{W})$, where $\Gamma$ is constraint language and the~weight range $\mathbb{W}$ is either $\mathbb{Z}$ or $\mathbb{N}$.
An instance (or formula) of~$CS(\Gamma, \mathbb{W})$ is a~set of~constraint applications from $\Gamma$ over a~\bmp{common} set of~variables, each application having a~weight from~$\mathbb{W}$.
\end{definition}

We denote the~number of~constraint applications in formula~$\Phi$ by $|\Phi|$ and the~sum of~absolute values of~all weights in $\Phi$ by $||\Phi||$.
For an assignment vector $x$, the~integer $\Phi(x)$ is the~sum of~weights of~the~constraints applications satisfied by $x$.

In the~decision problem  \textsc{Max CSP$(\Gamma, \mathbb{W}, c)$} we~are given a~formula~$\Phi$ from $\cs(\Gamma, \mathbb{W})$ over $n$~variables such that $||\Phi|| \leq n^c$, together with integer $t$, and we~ask if there is an assignment $x$ such that $\Phi(x) \ge t$.
We indicate the~parameter $c$ in order to be accurate about the~specific decision problems, for which we~can show hardness results.
When it does not lead to a~confusion, e.g., when some property holds for all $c$, we~refer to this family of~problems shortly as \textsc{Max CSP$(\Gamma, \mathbb{W})$}.
Whenever we~use the~$\Oh$-notation, we~do it with respect to a~fixed problem, that is, we~treat $\Gamma$ and $c$ as constants.

The most commonly studied case is expressed by $\mathbb{W} = \nn$~\cite{DeinekoJKK08, KhannaSTW00, MakarychevM17}, where the~weights can be interpreted as repetitions of~constant applications.
It is important to make this distinction because it can be the~case that \cspnn{} is polynomially solvable whereas \cspzz{} is NP-hard~\cite{JonssonK07}.
Although our main reduction framework works for $\mathbb{W} = \zz$, we~are able to transfer the~compression lower bounds to the~case $\mathbb{W} = \nn$ as long as   \cspnn{} is NP-hard.

Another decision problem that is related to constraint systems is \textsc{Exact CSP$(\Gamma, \mathbb{W})$}, where we~ask whether there is an assignment for which the~\bmp{satisfied weights} sum up exactly to a~given integer~\cite{LincolnWW18, Williams07}.
Even though we~focus on the~maximization variant, we~formulate our reductions so that they could be employed for other problems over constraint systems or larger weight domains.

\subsection{Types of~constraints} \label{sec:constraint:types}

We start by formally defining the~most important constraints and constraints properties.
They~allow us to formulate the~dichotomy theorem for \maxcsp.
We use the Boolean notation for negation, i.e., $\neg x = 1 - x$ for $x \in \{0, 1\}$.

\begin{itemize}
\itemsep0em 
\item A~constraint is trivial if it is either always 1 or always 0 regardless of~the~arguments.

\item T and F are unary constraints given by $T(x) = x$ and $F(x) = \neg x$.\bmpr{Remark that we~interchangeably use Boolean and 0/1-representation of~true and false.}

\item $\OR_k$ and $\AND_k$ are $k$-ary constraints, such that $\OR_k(x_1, \ldots, x_k) = \bigvee_{i=1}^k x_i$ and $\AND_k(x_1, \ldots, x_k) = \bigwedge_{i=1}^k x_i$.
The Not-All-Equal constraint is defined as
$\NAE_k(x_1, \ldots, x_k) = OR_k(x_1, \ldots, x_k) \land OR_k(\neg x_1, \ldots, \neg x_k)$.

\item $\XOR_k$ is a~$k$-ary constraint defined as $\XOR_k(x_1, \ldots, x_k) = x_1 +  \ldots + x_k \bmod 2$. We abbreviate
$\XOR = \XOR_2$.

\item A~constraint $f$ is 0-valid (resp. 1-valid) if $f(0, 0, \dots, 0) = 1$ (resp.  $f(1, 1, \dots, 1) = 1$).

\item A~constraint $f$ is 2-monotone if $f(x_1, x_2, \dots, x_k) = (x_{i_1} \land x_{i_2} \land \dots \land x_{i_p}) \lor (\neg x_{j_1} \land \neg x_{j_2} \land \dots \land \neg x_{j_q})$, for some $p, q \ge 0$, $(p,q) \neq (0,0)$, i.e., $f$ is equivalent to a~DNF-formula~with at most two terms: one
containing only positive literals and the~other containing only negative literals.

\item A~constraint $f$ is C-closed (complementation-closed) if for every assignment $x \in \{0,1\}^{\ar(f)}$, $f(x) = f(\bar x)$, where $\bar x$ stands for the~bit-wise complement of~$x$. 

\item A~constraint $f$ is symmetric if for any two assignments $x_1, x_2 \in \{0,1\}^{\ar(f)}$ having the~same number of~ones, it holds that $f(x_1) = f(x_2)$.
\end{itemize}

A~constraint language $\Gamma$ is called 0-valid, 1-valid, 2-monotone, C-closed, or symmetric, if all {non-trivial} constraints in $\Gamma$ satisfy the~respective property. We call~$\Gamma$ \emph{non-trivial} if it contains at least one non-trivial \bmp{constraint}.
This regime is convenient for formulating the~fundamental dichotomy theorem for Boolean \textsc{Max CSP}.
For our purposes it is only important that APX-hardness entails NP-hardness.

\begin{theorem}[{\cite[Theorem 2.11]{KhannaSTW00}, cf.~\cite{Creignou95}}]\label{thm:dichotomy}
\cspnn{} is solvable in polynomial time if $\Gamma$ is either 0-valid, 1-valid, or 2-monotone.
Otherwise, the~problem is APX-hard.
\end{theorem}

\subsection{Closures of~constraint languages}

\begin{definition}\label{def:expressible}
Let constraints $f, g$ be respectively $k$-ary and $d$-ary.
We say that $g$ is expressible by $f$ with constants if the~following identity holds
$$g(x_1, x_2, \dots, x_d) = f(\xi_1, \xi_2, \dots, \xi_k),$$
for a~vector $(\xi_1, \xi_2, \dots, \xi_k)$, where each $\xi_j$ is either a~variable $x_i$ for some $i \in [d]$ or one of~the~constants 0, 1.

We say that $g$ is expressible by $f$ with literals if such an identity holds for a~vector $(\xi_1, \xi_2, \dots, \xi_k)$, where each $\xi_j$ is a~literal: either a~variable $x_i$ or its negation $\neg x_i$ for some $i \in [d]$.
\end{definition}

For a~constraint language $\Gamma$ we~introduce its closures:

\begin{itemize}
\itemsep0em 
    \item the~language $\Gamma^{T,F}$ contains all functions expressible by $f \in \Gamma$ with constants,
    \item the~language $\Gamma^{LIT}$ contains all functions expressible by $f \in \Gamma$ with literals,  \item the~language $\Gamma^{NEG}$ is the~\emph{negation-wise closure} of~$\Gamma$, i.e., $\Gamma^{NEG} = \bigcup_{f \in \Gamma} \{f, \neg f\}$.
\end{itemize}
It is easy to see that the~closures satisfy $(\gtf)^{T,F} = \gtf$, $(\glit)^{LIT} = \glit$, $(\gneg)^{NEG} = \gneg$.
We will be particularly interested in those constraint languages \bmp{in which negated literals can be expressed,} as in, e.g., \textsc{$d$-CNF-SAT} or \textsc{$d$-NAE-SAT}.
These are the~languages that satisfy $\Gamma~= \glit$.
Below we~present examples on how to capture important CSPs using our definitions.
\begin{itemize}
\itemsep0em 
    \item \textsc{$d$-CNF-SAT} = \textsc{CSP$({\Gamma_{d\text{-SAT}}})$} for $\Gamma_{d\text{-SAT}} = \{\OR_d\}^{LIT}$,
    \item \textsc{$d$-NAE-SAT} = \textsc{CSP$(\{\NAE_d\}^{LIT})$},
    \item \textsc{Max E$d$-Lin} = \textsc{Max CSP$(\{\XOR_d\}^{NEG}, \nn)$},
    \item \textsc{Max Cut} = \textsc{Max CSP$(\{\XOR\}, \nn)$},
    \item \textsc{Max DiCut} = \textsc{Max CSP$(\{f\}, \nn)$} for $f(x_1, x_2) = x_1 \land \neg x_2$.
\end{itemize}

\subsection{Constraint implementations}

We describe a~technique that has been introduced in order to prove Theorem~\ref{thm:dichotomy}~\cite{KhannaSTW00}.
The idea~is to \emph{implement} a~constraint $f$ by a~collection of~other constraints, so that satisfying $f$ is equivalent to maximizing the~number of~satisfied constraints in that collection.
It allows to express formulas from \cspnn{1} by those from \cspnn{2}, as long constraints in $\Gamma_1$ can be implemented by those in $\Gamma_2$.

The caveat is that each implementation may introduce new auxiliary variables whereas for our purposes we~need reductions that increase the~number of~variables only by a~multiplicative constant.
Therefore the~reductions by Khanna~et al.~\cite{KhannaSTW00} do not transfer to the~compression paradigm and we~will use the~implementations in a~different way.
On the~other hand, our reductions do not preserve approximation factors.

\begin{definition}[{\cite[Definition 3.1]{KhannaSTW00}}]
A~collection of~unit-weighted constraint applications $C_1, C_2, \dots C_m$ over a~set
of~variables $x = \{x_1, x_2, \dots, x_p\}$ called primary variables and $y = \{y_1, y_2, \dots, y_q\}$ called auxiliary variables, is an $\alpha$-implementation of~a~constraint $f(x)$ for a~positive integer $\alpha$ if the~following conditions hold.

\begin{enumerate}
\itemsep0em 
\item For any assignment to $x$ and $y$, at most $\alpha$ constraint applications from $C_1, C_2, \dots C_m$ are satisfied.

\item For any $x$ such that $f(x) = 1$, there exists $y$ such that exactly $\alpha$ constraint applications are satisfied.

\item For any $x, y$ such that $f(x) = 0$, at most $\alpha~- 1$ constraint applications are satisfied.
\end{enumerate}
 An implementation is a~\emph{strict} $\alpha$-implementation if for every $x$ such that $f(x) = 0$, there exists $y$ such that exactly $\alpha~- 1$ constraint applications are satisfied.
\end{definition}

We say that a~constraint language $\Gamma$ (strictly) implements a~constraint $f$ if there exists a
(strict) $\alpha$-implementation of~$f$ using constraints of~$\Gamma$ for some constant $\alpha$.
We use $\Gamma~\Longrightarrow f$
to denote that $\Gamma$ implements $f$, and $\Gamma~\stackrel{s}{\Longrightarrow} f$ to denote that $\Gamma$ strictly implements $f$.
The above
notation is also extended to allow the~target to be a~family of~functions.

We omit the~machinery that was employed to construct implementations and we~will just exploit the~following results in a~black-box manner.
We are not going to rely on the~strictness property in further sections as we~need it only to ensure transitivity of~implementations for the~sake of~proving Lemma~\ref{lem:nphard-xor}, which was not stated explicitly in~\cite{KhannaSTW00}.

\begin{lemma}[{\cite[Lemma~3.5]{KhannaSTW00}}]\label{lem:strict-transitive}
If $\Gamma_1 \stackrel{s}{\Longrightarrow} \Gamma_2$ and $\Gamma_2 \stackrel{s}{\Longrightarrow} \Gamma_3$, then $\Gamma_1 \stackrel{s}{\Longrightarrow} \Gamma_3$.
\end{lemma}

\begin{lemma}[{\cite[Lemma~4.5]{KhannaSTW00}}]\label{lem:c-closed}
Let $f$ be a~non-trivial constraint which is C-closed and is not 0-valid (or equivalently not 1-valid). Then $\{f\} \stackrel{s}{\Longrightarrow} \XOR$.
\end{lemma}

\begin{lemma}[{\cite[Lemma~4.6]{KhannaSTW00}}]\label{lem:not-c-closed}
Let $f_0$, $f_1$, and $g$ be non-trivial constraints, possibly identical, which are not 0-valid,
not 1-valid, and not C-closed, respectively. Then  $\{f_0, f_1, g\} \stackrel{s}{\Longrightarrow} \{T, F\}$.
\end{lemma}

\begin{lemma}[{\cite[Lemma~4.11]{KhannaSTW00}}]\label{lem:not-c-closed-xor}
Let $f$ be a~constraint which is not 2-monotone. Then $\{f, T, F\} \stackrel{s}{\Longrightarrow} \XOR$.
\end{lemma}

\begin{lemma}\label{lem:nphard-xor}
Let $\Gamma$ be a~constraint language that is neither 0-valid, 1-valid, nor 2-monotone. Then $\Gamma~\stackrel{s}{\Longrightarrow} \XOR$.
\end{lemma}
\begin{proof}
If $\Gamma$ is C-closed, then the~claim follows from Lemma~\ref{lem:c-closed}.
Otherwise, we~can strictly implement $T, F$ (Lemma~\ref{lem:not-c-closed}) and then rely on transitivity (Lemma~\ref{lem:strict-transitive}) to implement $\XOR$ with Lemma~\ref{lem:not-c-closed-xor}.
\end{proof}

\subsection{Parameterized complexity} \label{sec:prelims:fpt}
A~\emph{parameterized problem} is a~decision problem in which every input has an associated positive integer \emph{parameter} that captures its complexity in some well-defined way. In our study of~CSPs we~use the~number of~variables as the~parameter, but other choices have been considered~\cite{GutinY17,KratschMW16,KratschW10}. For a~parameterized problem~$a~\subseteq \Sigma^* \times \mathbb{N}$, a~decision problem~$B \subseteq \Sigma^*$, and a~function~$f \colon \mathbb{N} \to \mathbb{N}$, a~\emph{compression} of~$A$ into~$B$ of~size~$f$ is an algorithm that, on input~$(x,k) \in \Sigma^* \times \mathbb{N}$, takes time polynomial in~$|x| + k$ and outputs an instance~$y \in \Sigma^*$ such that~$(x,k) \in A$ if and only if~$y \in B$, and such that~$|y| \leq f(k)$. A~\emph{kernelization} algorithm of~size~$f$ for problem~$A$ reduces any instance~$(x,k)$ to an~$f(k)$-sized equivalent instance of~the~\emph{same problem} in polynomial time. 


\section{Characteristic polynomials}\label{sec:polynomials}

In this section we~provide the~technique necessary for expressing one constraint system by another without introducing too many auxiliary variables.
This insight is based on interpreting constraints as multilinear polynomials.

\begin{definition}
For a~$k$-ary constraint \bmp{$f \colon \{0,1\}^k \to \{0,1\}$} its characteristic polynomial $P_f$ is the~unique $k$-ary multilinear polynomial over $\mathbb{R}$ satisfying $f(x) = P_f(x)$ for any $x \in \{0,1\}^k$.
\end{definition}


It is easy to construct such a~polynomial. \bmp{First define} 
$P_s(x_1, x_2, \dots, x_k) = \prod_{i=1}^k R_i^s(x_i)$ for a~vector $s \in \{0,1\}^k$,
where $R_i^s(x) = x$ if $s_i = 1$ and $R_i^s(x) = 1 - x$ otherwise.
Formally, $P_s$ is the~sequence of~coefficients one gets by expanding all parentheses.
It is easy to see that they are all integers.
It holds that $P_s(s) = 1$, \bmp{while} $P_s(x) = 0$ for any $x \neq s$.
For a~constraint $f$ its characteristic polynomial is given as
$P_f(x_1, x_2, \dots, x_k) = \sum_{s: \, f(s) = 1} P_s(x_1, x_2, \dots, x_k).$
It is known that no other multilinear polynomial can take identical values on $\{0,1\}^k$~\cite{NisanS94, Williams07}.
This also means we~can interchangeably analyze polynomials as formal objects and as functions on $\{0,1\}^k$.   

\begin{observation}
For any Boolean function~$f$, the~coefficients of~the~characteristic polynomial~$P_f$ are integers.
\end{observation}

Let $\deg (f) = \deg (P_f)$ and $\deg (\Gamma) = \max_{f \in \Gamma} \deg (f)$. For a~$k$-ary constraint $f$ we~refer to the~coefficient at the~unique $k$-ary monomial in $P_f$ as the~leading coefficient.
The leading coefficient is non-zero iff.~$\deg(P_f) = k$.

If $g$ is expressible by $f$ with literals, then we~can obtain $P_g$ from $P_f$ by replacing each literal with negation $\neg x_i$ by $1-x_i$ and expanding the~parenthesis within monomials.
If $g$ is expressible by $f$ with constants, then we~just substitute 0 or 1 for particular variables and remove monomials containing 0.
These transformations imply $\deg(\gtf) = \deg(\glit) = \deg(\gneg) = \deg(\Gamma)$.

As an example, consider \textsc{Max 3-NAE-SAT}.
The function $\NAE_3(x_1, x_2, x_3)$ has the~degree-2 characteristic polynomial $x_1 + x_2 + x_3 - x_1x_2 -x_1x_3 - x_2x_3$, which allows us to construct a~compression for \textsc{Max 3-NAE-SAT} of~size $\Oh(n^2\log n)$ by summing coefficients at all $\Oh(n^2)$ monomials.
On the~other hand, $\OR_2(x_1, x_2) = x_1 + x_2 - x_1x_2 = \NAE_3(x_1, x_2, 0)$, which indicates that solving \textsc{Max 3-NAE-SAT} should not be easier than \textsc{Max 2-CNF-SAT}.
We will formalize these arguments in the~next section.

We are now going to show that the~set of~characteristic polynomials of~all constraints expressible by $f$ with constants spans the~linear space of~multilinear polynomials over $\mathbb{Q}$ with degrees at most $\deg (f)$.
We first prove that this set contains polynomials of~all degrees up to $\deg (f)$ and then use them to express a~basis of~the~linear space.

\begin{lemma}\label{lem:express-g}
Let $f$ be a~$k$-ary constraint.
For any
$1 \le d \le \deg(f)$ there exists a~$d$-ary constraint $g$ expressible by $f$ with constants, such that its characteristic polynomial $P_g$ has degree exactly $d$, i.e., its leading coefficient is non-zero.
\end{lemma}
\begin{proof}
The degree of~$d$-ary $P_g$ can be at most $d$ so we~just need to show that the~leading coefficient would be non-zero.
We proceed by induction over $i = \deg(f) - d$.
Suppose that $d = \deg(f)$ and consider indices $j_1, \dots, j_d$ for which a~monomial $\prod_{i=1}^d x_{j_i}$ has a~non-zero coefficient~$\alpha$.
We define $g_0(x_1, x_2, \dots, x_d) = f(\xi_1, \xi_2, \dots, \xi_k)$, where $\xi_{j_i} = x_i$ and other $\xi_j$ are set to 0.
It is easy to verify that the~leading coefficient of~$P_{g_0}$ is~$\alpha$, which proves the~induction basis for $i=0$.

Suppose now the~induction hypothesis holds for $i-1$, that is,
$g_{d+1}(x_1, x_2, \dots, x_{d+1})$ is expressible by $f$ with constants and $\deg(P_{g_{d+1}}) = d+1$.
First consider the~case where $P_{g_{d+1}}$ has a~non-zero coefficient at some monomial of~degree $d$.
Let $j$ be the~unique index for which $x_j$ does not occur in this monomial.
Then $g_d(x_1, x_2, \dots, x_d) = g_{d+1}(x_1, \dots, x_{j-1}, 0 , x_{j}, \dots, x_{d})$ is expressible by $f$ with constants and has degree $d$.

Finally, suppose that $P_{g_{d+1}}$ has zero coefficients at all monomials of~degree $d$.
We define $g_d(x_1, x_2, \dots, x_d) = g_{d+1}(x_1, x_2, \dots, x_d, 1)$.
Since the~monomial $\prod_{i=1}^d x_i$ does not appear in $P_{g_{d+1}}$, the~leading coefficient in $P_{g_d}$ is the~same as in $P_{g_{d+1}}$.
\end{proof}

\begin{lemma}\label{lem:identity}
Let $f$ be a~non-trivial constraint and $P$ be a~multilinear polynomial over $\mathbb{Q}$ on $\ell$ variables of~degree $0 \le d \le deg(f)$.
There exists a~sequence of~constraint applications $\langle f_i, (j_i^1, \dots, j_i^{\ar(f_i)}), \alpha_i\rangle_{i=1}^M$ on $\ell$ variables, where each constraint $f_i$ is expressible by $f$ with constants, and $\alpha_i \in \mathbb{Q}$, such that the~following polynomial identity holds.

$$P(x_1, \dots, x_\ell) = \sum_{i=1}^M \alpha_i \cdot P_{f_i}(x_{j_i^1}, \dots, x_{j_i^{\ar(f_i)}}).$$
\end{lemma}

Before proving this claim, let us demonstrate it on a~less obvious example than the~one with $\NAE_3$ and $\OR_2$. 
Let $P(x_1, x_2, x_3)$ be the~characteristic polynomial of~the~constraint $\OR_3(x_1, x_2, \neg x_3)$:
\begin{align*}
P(x_1, x_2, x_3) &= x_1x_2x_3 + (1-x_1)x_2x_3 + x_1(1-x_2)x_3 + x_1x_2(1-x_3) + (1-x_1)x_2(1-x_3)  \\
&+ x_1(1-x_2)(1-x_3) + (1-x_1)(1-x_2)(1-x_3) = 1 - x_3 + x_1x_3 + x_2x_3 - x_1x_2x_3.
\end{align*}
We will represent it with characteristic polynomials from $\{\EX_3\}^{T,F}$, where $\EX_3(x_1, x_2, x_3) = 1$ iff. exactly one variable is 1. Its characteristic polynomial $Q$ is given as:
\begin{align*}
Q(x_1, x_2, x_3) &= x_1(1-x_2)(1-x_3) + (1-x_1)x_2(1-x_3) + (1-x_1)(1-x_2)x_3 \\
&= x_1 + x_2 + x_3 -2x_1x_2 -2x_1x_3  -2x_2x_3 +3x_1x_2x_3.
\end{align*}
We can express $P$ as the~following linear combination where, e.g., $Q(x_1, x_2, 0)$ is the~characteristic polynomial for $\EX_3(x_1, x_2, 0)$, which is a~binary constraint expressible by $\EX_3$ with constants:
\begin{align*}
P(x_1, x_2, x_3) =& -\frac{1}{3}Q(x_1, x_2, x_3) +\frac{1}{3}Q(x_1, x_2, 0) -\frac{1}{6}Q(x_1, x_3, 0) -\frac{1}{6}Q(x_2, x_3, 0) \\
& + \frac{1}{6}Q(x_1, 0, 0) +\frac{1}{6}Q(x_2, 0, 0) - \frac{1}{3}Q(x_3, 0, 0) + Q(1,0,0).
\end{align*}

\begin{proof}[Proof~of~Lemma~\ref{lem:identity}]
%
%

We proceed by induction over the~degree of~$P$.
Since $f$ is non-trivial, it admits a~satisfying assignment $s_T$.
If $P$ is constant, then $P(x) = \alpha~\cdot f(s_T)$ for some $\alpha~\in \mathbb{Q}$.

Suppose now $d = \deg(P) \ge 1$. For each~$S \in \binom{[\ell]}{d}$ let $\alpha_S$ denote the~(potentially zero) coefficient in $P$ at the~monomial $\prod_{i \in S} x_i$.
By Lemma~\ref{lem:express-g} there is a~$d$-ary constraint $f_d$, which is expressible by $f$ with constants and $\deg (P_{f_d}) = d$.
Let $\beta_d$ denote the~leading coefficient of~$P_{f_d}$.
The polynomial

$$P'(x_1, \dots, x_\ell) = P(x_1, \dots, x_\ell) - \sum_{\substack{\{i_1, \ldots, i_d\} = S \in \binom{[\ell]}{d} \\ i_1 < \ldots < i_d}} \frac{\alpha_S}{\beta_d} \cdot P_{f_d}(x_{i_1}, \dots, x_{i_d})$$
has no monomials of~degree $d$, since each term in the~sum subtracts exactly one of~them.
\mic{The polynomial $P'$ has degree at most $d-1$,
so we~can apply the~induction hypothesis to it
and represent $P'$ as~a~linear combination of~characteristic polynomials of~constraints from $\{f\}^{T,F}$.
We~obtain the~claim by adding these polynomials to the~sum above.}
\end{proof}

Since ${f_i}$ and $P_{f_i}$ coincide as functions on $\{0,1\}^{\ar(f_i)}$, Lemma~\ref{lem:identity} allows us to represent any constraint of~degree at most $\deg(f)$ as a~linear combination of~constraints from $\{f\}^{T,F}$.

\begin{proposition}\label{pro:identity}
Let $g, f$ be constraints, such that $f$ is non-trivial and $\deg(g) \le \deg(f)$.
There exists a~sequence of~constraint applications $\langle f_i, (j_i^1, \dots, j_i^{\ar(f_i)}), \alpha_i\rangle_{i=1}^M$ on $\ar(g)$ variables, where each constraint $f_i$ is expressible by $f$ with constants, and $\alpha_i \in \mathbb{Q}$, such that the~following identity holds for all binary vectors.

$$g(x_1, \dots, x_{\ar(g)}) = \sum_{i=1}^M \alpha_i \cdot f_i(x_{j_i^1}, \dots, x_{j_i^{\ar(f_i)}}).$$
\end{proposition}

This resembles the~idea~of~implementation, where we~are additionally equip constraints with (potentially negative) rational weights, but in return it does not require introducing any auxiliary variables.

\section{Reductions between constraint systems}
\label{sec:reductions}

We first formalize our notion of~reduction.
The objects we~work with are the~constraint systems and the~reductions between them imply analogous relations between the~associated decision problems.
The reduction is crafted \bmp{in such a~way that it} preserves the~numbers of~variables and constraints up to a~constant factor, and the~total weight up to a~polynomial factor.

\begin{definition}\label{def:transform}
A~linear transformation from a~constraint system $CS(\Gamma_1, \mathbb{W}_1)$
 to another constraint system $CS(\Gamma_2, \mathbb{W}_2)$ is a~polynomial-time procedure that given a~formula~$\Phi_1$ of~$CS(\Gamma_1, \mathbb{W}_1)$ over $n_1$~variables and integer $t_1$, returns a~formula~$\Phi_2 \in CS(\Gamma_2, \mathbb{W}_2)$ over $n_2$ variables and integer $t_2$, so that the~following conditions hold:
\begin{enumerate}
\itemsep0em 
    \item $n_2 = \Oh(n_1)$,
    \item $|\Phi_2| = \Oh(|\Phi_1| + n_1)$,
    \item $||\Phi_2|| \le ||\Phi_1|| \cdot n_1^{\Oh(1)}$, 
    \item $\exists_x \Phi_1(x) = t_1 \Longleftrightarrow \exists_y \Phi_2(y) = t_2$,
    \item $\exists_x \Phi_1(x) \ge t_1 \Longleftrightarrow \exists_y \Phi_2(y) \ge t_2$.
\end{enumerate}

If there exists a~linear transformation from $CS(\Gamma_1, \mathbb{W}_1)$ to $CS(\Gamma_2, \mathbb{W}_2)$, we~write concisely
$CS(\Gamma_1, \mathbb{W}_1) \le_{LIN} CS(\Gamma_2, \mathbb{W}_2)$.

If the~condition (1) can be replaced with a~stronger one: $n_2 = n_1 + \Oh(1)$, we~call the~transformation additive and write $CS(\Gamma_1, \mathbb{W}_1) \le_{ADD} CS(\Gamma_2, \mathbb{W}_2)$.

\end{definition}

\mic{Before moving forward, let us explain the~importance of~linear and additive transformations.
We formulate two claims, which
follow from the~properties in Definition~\ref{def:transform}.}

\begin{proposition}\label{pro:compression}
If  $CS(\Gamma_1, \mathbb{W}_1) \le_{LIN} CS(\Gamma_2, \mathbb{W}_2)$ and \textsc{Max CSP}$(\Gamma_2, \mathbb{W}_2, c)$ admits a~compression of~size $\Oh(n^d)$, then \textsc{Max CSP}$(\Gamma_1, \mathbb{W}_1, c - \Oh(1))$ also admits a~compression of~size $\Oh(n^d)$.
\end{proposition}

\begin{proposition}\label{pro:exp}
If $CS(\Gamma_1, \mathbb{W}_1) \le_{ADD} CS(\Gamma_2, \mathbb{W}_2)$ and \textsc{Max CSP}$(\Gamma_2, \mathbb{W}_2, c)$ admits an~algorithm with running time $T(n)$, then \textsc{Max CSP}$(\Gamma_1, \mathbb{W}_1, c - \Oh(1))$ admits an~algorithm with running time $T(n + \Oh(1))$.
\end{proposition}

In particular, additive transformations preserve running \bmp{times} of~the~form $2^{(1-\eps)n}n^{\Oh(1)}$.

\begin{lemma}
Linear transformations (resp. Additive transformations) are transitive.
\end{lemma}
\begin{proof}
Suppose $CS(\Gamma_1, \mathbb{W}_1) \le_{LIN} CS(\Gamma_2, \mathbb{W}_2) \le_{LIN} CS(\Gamma_3, \mathbb{W}_3)$ and we~want to show that  $CS(\Gamma_1, \mathbb{W}_1) \le_{LIN} CS(\Gamma_3, \mathbb{W}_3)$. 
Since $n_2 = \Oh(n_1)$ and $n_3 = \Oh(n_2)$, then $n_3 = \Oh(n_1)$.
To ensure properties (2, 3) are preserved, note that
$|\Phi_3| = \Oh(|\Phi_2| + n_2) = \Oh(|\Phi_1| + n_1)$ and
$||\Phi_3|| \le ||\Phi_1|| \cdot n_1^{\Oh(1)} \cdot n_2^{\Oh(1)} = ||\Phi_1|| \cdot n_1^{\Oh(1)}$.
The properties (4, 5) are equivalences that are transitive.
Moreover, if we~replace relation $\le_{LIN}$ with $\le_{ADD}$, then we~have that $n_3 = n_1 + \Oh(1)$ and the~other properties follow since an additive transformation is always a~linear one.
\end{proof}

We continue with two simple additive transformations, which will allow us to use negations of~constraints as an alternative to setting negative weights.\bmpr{`as an alternative to setting negative weights'?}

\begin{lemma}\label{lem:signed-reductions}
For every constraint language $\Gamma$ we~have
\begin{enumerate}
\itemsep0em
\item $\cs(\gneg, \zz) \le_{ADD} \cs(\Gamma, \zz)$,
    \item $\cs(\gneg, \zz) \le_{ADD} \cs(\gneg, \nn)$.
     
\end{enumerate}
\end{lemma}
\begin{proof}
Suppose we~are given $\Phi_1 \in \cs(\gneg, \zz)$ and integer $t$.
Observe that $(\neg f)(x) = 1 - f(x)$.
\mic{
For each negated constraint~$\neg f$ we~can thus remove each constraint application $\langle \neg f, (i_1, i_2, \dots i_k), w \rangle$, replace it with $\langle f, (i_1, i_2, \dots i_k), -w \rangle$, and subtract~$w$ from the~target weight~$t$.
We iterate this operation until all applied constraints belong to $\Gamma$.
We obtain
a~new formula~$\Phi_2$ of~\cs($\Gamma, \zz$) over the~same set of~variables and a~new threshold $t'$, so that $\Phi_2(x) = t' \Longleftrightarrow \Phi_1(x) = t$ (and the~same holds with '$=$' replaced by '$\ge$').}
This proves (1).

\bmp{For claim (2)}, consider a~formula~ $\Phi_1$ of~\cs($\gneg, \zz$).
As before, we~have $ f(x) = 1 - (\neg f)(x)$ and we~can replace each constraint application of~$f$ having a~negative weight by $(\neg f)(x)$ with the~opposite weight, and subtract the~weight from $t$.
By erasing all negative weights, we~obtain a~formula~$\Phi_2$ of~\cs($\gneg, \nn$)
over the~same set of~variables and new threshold $t'$, so that the~new instance is equivalent.
\end{proof}

The rest of~this section contains four lemmas
that form a~chain of~reductions from $\cs(\Gamma_1, \zz)$ to $\cs(\Gamma_2, \nn)$\bmp{, which is} valid as long as $\deg(\Gamma_1) \le \deg(\Gamma_2)$ and \cspnn{2} is NP-hard.
The proofs are based on the~frameworks of~characteristic polynomials and constraint implementations.

\begin{lemma}\label{lem:apply-poly}
Let $\Gamma_1, \Gamma_2$ be {non-trivial}
constraint languages satisfying
$\deg(\Gamma_1) \le \deg(\Gamma_2)$. Then \\ $\cs(\Gamma_1, \zz) \le_{ADD} \cs(\gtf_2, \zz)$.
\end{lemma}
\begin{proof}
Let $f \in \Gamma_2$ be a~constraint of~degree $\deg(\Gamma_2)$.
By Proposition~\ref{pro:identity} we~can represent each $g \in \Gamma_1$ as a~linear combination of~constraint applications $\langle f_i, (j_i^1, \dots, j_i^{\ar(f_i)}), \alpha_i\rangle$, where 
$\alpha_i \in \mathbb{Q}$,
$(j_1, \dots, j_{\ar(f_i)}) \in \binom{[\ar(g)]}{\ar(f_i)}$, and
functions $f_i$ are expressible by $f$ with constants.

Since $\Gamma_1$ is finite, there is a~common denominator of~all rational numbers occurring in these linear combinations. 
Therefore, we~can find a~single positive integer $\beta$, so that for each $g \in \Gamma_1$, the~function $\beta\cdot g$ can be represented as a~combination of~constraints from $\{f\}^{T,F}$ with integer coefficients. 



Given a~formula~$\Phi_1 \in \cs(\Gamma_1, \zz)$ we~replace each constraint application $\langle g, (j_1, \dots, j_k), w \rangle$ with a~set of~constraints applications from $\{f\}^{T,F}$ with these integer weights multiplied by~$w$.
We obtain a~formula~$\Phi_2 \in \cs(\gtf_2, \zz)$ over the~same set of~variables, such that $\Phi_2(x) = \beta~\cdot \Phi_1(x)$.
This implies the~properties (4, 5) from Definition~\ref{def:transform}.
We have not increased the~number of~variables, so the~transformation is additive, and $|\Phi_1|, ||\Phi_1||$ have been increased by a~constant factor depending only on~$\Gamma_1$ and~$\Gamma_2$.
\end{proof}

\mic{Consider two formulas $\Phi_1, \Phi_2$ over sets of~variables $V_1, V_2$, respectively, which might have a~non-empty intersection.
We define the~sum of~these formulas, $\Phi_1 + \Phi_2$, over the~set of~variables $V_1 \cup V_2$ by taking a~union of~their sets of~constraint applications and merging pairs of~applications that share the~same constraint and the~same tuple of~variables, i.e., replacing the~pair with a~single application with a~weight being the~sum of~the~respective weights.
For an integer $\alpha$, the~formula~$\alpha\cdot\Phi$ has the~same constraint applications as $\Phi$, but with weights multiplied by $\alpha$.}\bmpr{Nice!}

\begin{lemma}\label{lem:implement-tf}
Let $\Gamma$ be
a~non-trivial constraint language, which is neither 0-valid nor 1-valid.
Then \\ $\cs(\gtf, \zz) \le_{ADD} \cs(\Gamma, \zz)$.
\end{lemma}
\begin{proof}
Given a~formula~$\Phi \in \cs(\gtf, \zz)$, we~add two auxiliary variables $x_T, x_F$ and we~translate each constraint application $\langle \hat f, (j_1, \dots, j_k), w \rangle$, where $\hat f$ is expressible by $f \in \Gamma$ with constants, into an application of~$f$ by replacing constants 0, 1 with variables $x_T, x_F$.
\mic{Let us refer to this formula~as~$\Phi_1 \in \cs(\Gamma, \zz)$ and note that $|\Phi_1| = |\Phi|$ and $||\Phi_1|| = ||\Phi||$.}

In the~next step, we~will use implementations to impose particular conditions on $x_T, x_F$.
We refer to $x_T, x_F$ and all the~new variables introduced within the~implementations as auxiliary.
Let $a~= \Oh(1)$ denote their number.
In the~new formula~we~assume that the~first $n$ variables $x_1, \dots, x_n$ are the~primary variables and $x_{n+1}, \dots, x_{n+A}$ are auxiliary.
For $x \in \{0, 1\}^{n+A}$ let $x|_n$ stand for the~projection on the~first $n$ coordinates.

Assume first that $\Gamma$ is not C-closed.
Then it contains non-trivial functions $f_0$, $f_1$, and $g$, possibly identical, which are not 0-valid,
not 1-valid, and not C-closed, respectively.
By Lemma~\ref{lem:not-c-closed} we~know that $\Gamma$
$\alpha_1$-implements function $T$ and $\alpha_2$-implements function $F$ for some integers $\alpha_1, \alpha_2$.
Let $\alpha~= \alpha_1 + \alpha_2$.

\mic{We implement constraint applications $T(x_T)$ and $F(x_F)$,
that is, we~construct formulas $\Phi_T, \Phi_F \in \cs(\Gamma, \nn)$ over the~set of~auxiliary variables, such that satisfying $\alpha$ constraint applications in  $\Phi_T + \Phi_F$ is only possible when $x_T = 1$ and $x_F = 0$.
Let $W = 2 \cdot ||\Phi|| + 1$.
We define $\Phi_2 = \Phi_1 + W\cdot\Phi_T + W\cdot\Phi_F$,
that is, we~copy all the~constraint applications from $\Phi_1$ and add the~applications from $\Phi_T + \Phi_F$ with weights multiplied by $W$.}
Recall that we~have $(-||\Phi||) \le \Phi(x) \le ||\Phi||$ for all~$x$. 
By the~definition of~implementation, any assignment to $\Phi_2$ which does not satisfy $x_T = 1$ or $x_F = 0$ has value at most $(\alpha~- 1) \cdot W + ||\Phi|| < \alpha~W - ||\Phi||$.
If the~assignment $x$ satisfies $x_T = 1,\, x_F = 0$, it holds that $\Phi_2(x) = \alpha~W + \Phi(x|_n) \ge \alpha~W - ||\Phi||$.

Now, if $\Gamma$ is C-closed, it contains  a~ non-trivial  constraint  which  is  C-closed  and  not  0-valid.
Due~to~Lemma~\ref{lem:c-closed}, $\Gamma$
$\alpha$-implements function $\XOR$ for some constant $\alpha$.
\mic{We implement $\XOR(x_T, x_F)$,
that is, we~construct formula~$\Phi_{\XOR} \in \cs(\Gamma, \nn)$ over the~set of~auxiliary variables, such that satisfying $\alpha$ constraint applications in $\Phi_{\XOR}$ is only possible when $\XOR(x_T, x_F) = 1$.
As before, we~define $\Phi_2 = \Phi_1 + W\cdot\Phi_{\XOR}$, where
$W = 2 \cdot ||\Phi|| + 1$. }
Any assignment to $\Phi_2$ which does not satisfy $\XOR(x_T, x_F)$ has value at most $(\alpha~- 1) \cdot W + ||\Phi|| < \alpha~W - ||\Phi||$.
For an assignment $x$ satisfying $x_T = 1,\, x_F = 0$, it holds that $\Phi_2(x) = \alpha~W + \Phi(x|_n) \ge \alpha~W - ||\Phi||$.
It might also be the~case that $x_T = 0,\, x_F = 1$\bmp{. Then, }by C-closedness we~have $\Phi_2(x) = \Phi_2(\bar{x}) = \alpha~W + \Phi(\bar{x}|_n) \ge \alpha~W - ||\Phi||$, where $\bar{x}$ is the~bit-wise complement of~$x$.

We summarize the~transformation properties for both considered cases: the~new number of~variables is $n + \Oh(1)$, $|\Phi_2| = |\Phi| + \Oh(1)$, and $||\Phi_2|| = ||\Phi|| \cdot \Oh(1)$.
\mic{If $t < - ||\Phi||$, then we~know that all assignments $x$ satisfy $\Phi(x) > t$ and in such case we~could return an~empty formula~over a~singleton variable set (so the~only possible value is 0) and set threshold $t' = -1$: this is an~equivalent instance.}
To see properties (4, 5) observe that, assuming $t \ge - ||\Phi||$, $\Phi(x) = t$ (resp. $\Phi(x) \ge t$) holds for some assignment $x$ iff.~$\Phi_2(y) = t'$ (resp. $\Phi_2(y) \ge t'$) holds for some assignment $y$, where $t' = \alpha~W + t$.
\end{proof}

\begin{corollary}\label{cor:additive}
Let $\Gamma_1, \Gamma_2$ be {non-trivial} constraint languages such that
that $\deg(\Gamma_1) \le \deg(\Gamma_2)$ and $\Gamma_2$ is neither 0-valid nor 1-valid.
Then $\cs(\Gamma_1, \zz) \le_{ADD} \cs(\Gamma_2, \zz)$.
\end{corollary}

So far we~have established a~relation between $\Gamma_1$ and $\Gamma_2$, which allows us to transform one constraint system to another by adding only a~constant number of~new variables.
However, it~works only when we~allow negative weights.
The next two lemmas explain how to get rid of~negative weights, so that the~hardness results can be applied to the~natural setting with only non-negative weights.

\begin{lemma}\label{lem:unsigned-lit}
For every non-trivial constraint language $\Gamma$ it holds that 
$\cs(\Gamma, \zz) \le_{ADD} \cs(\glit, \nn)$.
\end{lemma}
\begin{proof}
For a~$k$-ary constraint $f \in \Gamma$ and $S \subseteq [k]$, let $f^S$ be the~function expressible by $f$ with literals given by negating variables with indices in $S$.
For example $\OR_3^{\{2\}}(x_1, x_2, x_3) = x_1 \lor \neg x_2 \lor x_3$.

Suppose we~are given $\Phi \in \cs(\Gamma, \zz)$ on $n$ variables.
Let $\mathcal{J}_f$ denote the~family of~all tuples $(j_1, \dots, j_k) \in [n]^k$ that appear in some constraint application of~$k$-ary constraint $f$ in $\Phi$.
For each $f \in \Gamma$, we~construct a~formula~$\Phi_f \in \cs(\glit, \nn)$ over the~same set of~variables as $\Phi$.
The formula~$\Phi_f$ consists of~all constraint applications of~the~form $\langle f^S, (j_1, \dots, j_k), 1 \rangle$ for all $S \subseteq [k]$ and $(j_1, \dots, j_k) \in \mathcal{J}_f$.

We claim that $\Phi_f$ is a~constant formula, i.e., $\Phi_f(x)$ does not depend on $x$.
For any assignment $x$ and $k$-tuple $(j_1, \dots, j_k)$, the~number of~sets $S$ for which $f^S(x_{j_1}, \dots, x_{j_k}) = 1$ equals exactly $|f| = |\{s \in \{0,1\}^k : f(s) = 1\}|$.
Therefore $\Phi_f(x) = |\mathcal{J}_f| \cdot |f|$ for any assignment $x$.

Let $W$ be the~absolute value of~the~minimum negative weight in $\Phi$ (if $\Phi$ has no negative weights, we~set $W = 0$).
\mic{We define $\Phi_1 = \Phi + W\cdot\sum_{f \in \Gamma} \Phi_f$.} 
Every constraint application from $\Phi$ appears in some $\Phi_f$, so after summing the~weights we~end up with only non-negative numbers.
Hence $\Phi_1 \in \cs(\glit, \nn)$ and $\Phi_1(x) = \Phi(x) +  W \cdot \sum_{f \in \Gamma}|\mathcal{J}_f| \cdot |f|$.
The second summand does not depend on $x$ and we~can set $t' = t +  W \cdot \sum_{f \in \Gamma}|\mathcal{J}_f| \cdot |f|$.
Since $ \sum_{f \in \Gamma}|\mathcal{J}_f| \cdot |f| = \Oh(|\Phi|)$,
we~have obtained an~additive transformation with $|\Phi_1| = \Oh(|\Phi|)$ and $||\Phi_1|| = \Oh(||\Phi||\cdot |\Phi|) = ||\Phi||\cdot n^{\Oh(1)}$.
\end{proof}

\begin{lemma}\label{lem:implement-lit}
Suppose non-trivial $\Gamma$ is neither 0-valid, 1-valid, nor 2-monotone.
Then $\cs(\glit, \nn) \le_{LIN} \cs(\Gamma, \nn)$.
\end{lemma}
\begin{proof}
We proceed similarly to Lemma~\ref{lem:implement-tf}.
Given a~formula~$\Phi \in \cs(\glit, \nn)$ on~$n$ variables,
for each variable $x_i$ we~introduce its negation-copy, that is, a~new variable $\hat{x_i}$, and translate each constraint application $\langle \hat f, (j_1, \dots, j_k), w \rangle$, where $\hat f$ is expressible by $f \in \Gamma$ with literals, into an application of~$f$ by replacing literals $x_i, \neg x_i$ with variables $x_i, \hat{x_i}$.
\mic{Let us refer to this formula~on $2n$ variables as~$\Phi_1 \in \cs(\Gamma, \nn)$ and note that $|\Phi_1| = |\Phi|$ and $||\Phi_1|| = ||\Phi||$.}

By Lemma~\ref{lem:nphard-xor} we~know that $\Gamma$
$\alpha$-implements $\XOR$ for same constant $\alpha$.
\mic{We implement $\XOR(x_i, \hat{x_i})$, that is, for each $i \in [n]$ we~construct a~formula~$\Phi_i \in \cs(\Gamma, \nn)$ over $x_i, \hat{x_i}$ and $\Oh(1)$ auxiliary variables (the~sets of~auxiliary variables are disjoint for different $i$), so that $\Phi_i$ can have value $\alpha$ only when $\XOR(x_i, \hat{x_i}) = 1$.
We define $\Phi_2 = \Phi_1 + W\cdot\sum_{i=1}^n \Phi_i$, where $W = ||\Phi|| + 1$.}
The total number of~variables is $\Oh(n)$ and we~have added $\Oh(n)$ new constraint applications.
As before, let $x|_n$ denote the~projection of~the~enlarged vector of~variables to the~original one.

Since the~weights are non-negative, we~have $0 \le \Phi_1(x) = \Phi(x|_n) \le ||\Phi||$ for all~$x$.
By the~definition of~implementation, any assignment to $\Phi_2$ which does not satisfy $\XOR(x_i, \hat{x_i})$ for some $i \in [n]$ has value at most $(n\alpha-1) \cdot W + ||\Phi|| < n\alpha~ W$.
Otherwise, we~have $\Phi_2(x) = n\alpha~W + \Phi(x|_n) \ge n\alpha~W$.

To summarize, the~new number of~variables is $\Oh(n)$, $|\Phi_2| = |\Phi| + \Oh(n)$, and $||\Phi_2|| = ||\Phi|| \cdot \Oh(n)$.
To see properties (4, 5) observe that we~can assume $t \ge 0$ \mic{(by the~same argument as in Lemma~\ref{lem:implement-tf})} and then $\Phi(x) = t$ (resp. $\Phi(x) \ge t$) holds for some assignment $x$ iff.~$\Phi_2(y) = t'$ (resp. $\Phi_2(y) \ge t'$) holds for some assignment $y$, where $t' = n \alpha~W + t$.
\end{proof}

In particular, all the~transformation above are linear, therefore \bmp{by transitivity} we~can summarize them as the~following corollary.

\begin{corollary}\label{cor:linear}
Let $\Gamma_1, \Gamma_2$ be {non-trivial} constraint languages such that
that $\deg(\Gamma_1) \le \deg(\Gamma_2)$ and $\Gamma_2$ is neither 0-valid, 1-valid, nor 2-monotone.
Then $\cs(\Gamma_1, \zz) \le_{LIN} \cs(\Gamma_2, \nn)$.
\end{corollary}

\section{Consequences for compression}
\label{sec:compression}

Having an upper bound on $\deg(\Gamma)$ already provides compression for \bmp{\cspzzc{}}, since we~can represent all constraint applications as polynomials, sum the~coefficients at all $\Oh\big(n^{\deg(\Gamma)}\big)$ monomials, and store the~weights in $\Oh(\log n)$ bits.
However, we~would like to compress the~given instance into an equivalent instance of~the~same problem.

\begin{theorem}\label{thm:csp-ub}
\cspnnc{} parameterized by the~number of~variables~$n$ admits a~compression of~size $\Oh(n^d \log n)$ for all $c$, where $d = \deg (\Gamma)$. Furthermore, there is a~polynomial-time algorithm that reduces any instance of~\cspnnc{} to an equivalent instance of~\maxcspg{\Gamma, \mathbb{N}, c + \Oh(1)} of~size~$\Oh(n^d \log n)$. The analogous statements for \cspzzc{} also hold.
\end{theorem}
\begin{proof}
If $\Gamma$ is either trivial, 0-valid, 1-valid, or 2-monotone, then \cspnnc{} can be solved in polynomial time, so the~kernelization is trivial.
Suppose that it does not have any of~these properties.
We will prove both claims by compressing a~formula~$\Phi_1 \in \cs(\Gamma, \zz)$ on $n$ variables into a~formula~$\Phi_2 \in \cs(\Gamma, \nn)$ satisfying $|\Phi_2| = \Oh(n^d)$.

First we~interpret each constraint application in $\Phi_1$ as a~polynomial of~degree at most $d$.
After summing these terms, we~obtain a~polynomial $P$ with $\Oh(n^d)$ monomials, each associated with an integer weight of~absolute value bounded by $||\Phi_1||$.
A~monomial $\Pi_{i=1}^k x_i$ is the~characteristic polynomial for the~constraint $\text{AND}_k(x_1, \dots, x_k)$, therefore $P$ can be treated as a~formula~of~$\cs(\Gamma_{d\text{-AND}}, \zz)$ for $\Gamma_{d\text{-AND}}$ being the~language consisting of~functions $\text{AND}_k$ for all $k \le d$.
Since $\deg(\Gamma_{d\text{-AND}}) = d$, we~can apply Corollary~\ref{cor:linear} to obtain an equivalent formula~$\Phi_2 \in \cs(\Gamma, \nn)$ on $\Oh(n)$ variables, such that $|\Phi_2| = |\Phi_1| \cdot \Oh(1) + \Oh(n) = \Oh(n^d)$ and $||\Phi_2|| = ||\Phi_1|| \cdot n^{\Oh(1)}$, so each weight can be stored in $\Oh(\log n)$ bits.
\end{proof}

The self-reduction in Theorem~\ref{thm:csp-ub} is almost, but not quite, a~kernelization: the~formal decision problem we~reduce to is not the~same as the~original one, due to the~increase in weight values. Theorem~\ref{thm:csp-ub} formalizes Theorem~\ref{thm:upperbound:informal}.

Our lower bounds are based on reducing \textsc{Max $d$-CNF-SAT} to \cspnn{} for $\deg(\Gamma) = d$.
For $d \ge 3$ it is known that even \bmp{the} non-maximization variant \textsc{$d$-CNF-SAT} does not admit \bmp{a} compression of~size $\Oh\left(n^{d-\eps}\right)$~\cite{DellM14}.
However, the~\textsc{$2$-CNF-SAT} problem is solvable in polynomial time and only its maximization version becomes NP-hard.
We first note that \textsc{Max $2$-CNF-SAT} also cannot have any non-trivial compression.

\begin{lemma}\label{lem:2sat}
\textsc{Max CSP}$(\Gamma_{2\text{-SAT}}, \nn, 3)$ does not admit a~compression of~size  $\Oh\left(n^{2-\eps}\right)$ for any $\eps > 0$, unless $\text{NP} \subseteq \text{co-NP} / \text{poly}$. 
\end{lemma}
\begin{proof}
We reduce from the~\textsc{Vertex Cover} problem parameterized by the~number of~vertices~$n$. Here, we~are given an undirected graph $G = (V,E)$ with integer $k$, and we~are asked whether one can choose $k$ vertices so that each edge is incident to at least one of~them.
\textsc{Vertex Cover} is known not to admit an $\Oh\left(n^{2-\eps}\right)$-size compression, unless $\text{NP} \subseteq \text{co-NP} / \text{poly}$~\cite{DellM14}.

Given graph $G$, we~construct a~\textsc{2-SAT} formula~as follows.
We create a~variable $x_v$ for each vertex $v \in V$ and add a~constraint \bmp{application} 
$\OR_2(\neg x_v, \neg x_v)$ with a~unitary weight.
We also create a~constraint \bmp{application} 
$\OR_2(x_u, x_v)$ with weight $2n$ for each edge $uv \in E$, so the~total weight of~the~instance is $\Oh(n^3)$.
Now, if $S \subseteq V$ is a~vertex cover in $G$, then the~assignment setting $x_v = 1$ for $v \in S$ has value $2n\cdot |E| + n - |S|$.
On the~other hand, if some assignment has value at least $2n\cdot |E|$, then it must satisfy all the~edge-constraints and $\{v \in V : x_v = 1\}$ forms a~vertex cover.
Therefore there is an~assignment of~value at least $2n\cdot |E| + k$ iff.~$G$ has a~vertex cover of~size at most $n - k$. 

As the~parameter is preserved by the~reduction, a~compression of~size~$\Oh(n^{2-\eps})$ for \textsc{Max CSP}$(\Gamma_{2\text{-SAT}}, \nn, 3)$ implies the~same for \textsc{Vertex Cover}, which proves the~lemma.
\end{proof}

\begin{theorem}\label{thm:csp-nn-lb}
Let non-trivial $\Gamma$ be neither 0-valid, 1-valid, nor 2-monotone, and let $d = \deg (\Gamma)$.
Then there is a~constant $c$ such that \cspnnc{} does not admit a~compression of~size  $\Oh\left(n^{d-\eps}\right)$ for any $\eps > 0$, unless $\text{NP} \subseteq \text{co-NP} / \text{poly}$.
\end{theorem}
\begin{proof}
\mic{First observe that a~characteristic polynomial of~a~constraint with $d=1$ is linear.
Since this polynomial is 0/1-valued, it can depend only on one variable, which means that the~constraint is 2-monotone.
Hence we~can assume that $d \ge 2$.}

We know that there is a~constant $c_d$ so that \textsc{Max CSP}$(\Gamma_{d\text{-SAT}}, \nn, c_d)$ does not admit compression of~size $\Oh\left(n^{d-\eps}\right)$ unless $\text{NP} \subseteq \text{co-NP} / \text{poly}$: for $d=2$ it is due to Lemma~\ref{lem:2sat} and for $d \ge 3$ it follows from the~compressibility hardness for the~non-maximization variant of~$d$-CNF-SAT~\cite{DellM14}.
As noted in Proposition~\ref{pro:compression}, if we~had such a~compression for \cspnnc{} with sufficiently large $c$, then by Corollary~\ref{cor:linear} it would transfer to \textsc{Max CSP}$(\Gamma_{d\text{-SAT}}, \nn, c_d)$.
\end{proof}

Theorem~\ref{thm:csp-nn-lb} is a~formalization of~Theorem~\ref{thm:lowerbound:informal} in the~introduction.
Below we~present some important corollaries from it.

\paragraph{Applications to specific CSPs} Having established both the~lower and upper bound, we~can refer to $\deg(\Gamma)$ as the~\emph{optimal compression exponent} for \cspnn{}.
We are now equipped with a~handy but powerful tool for determining the~optimal compressibility of~\cspnn{} as this task reduces to computing the~degrees of~characteristic polynomials in $\Gamma$.

As an example, we~apply this technique to compute the~{optimal compression exponent} for \textsc{Max $k$-NAE-SAT}, \textsc{Max E$k$-Lin}, and \textsc{Max $k$-Exact-SAT}, which are all NP-hard for $k \ge 2$.
We have
\begin{itemize}
\itemsep0em 
    \item \textsc{Max $k$-NAE-SAT} = \textsc{Max CSP$(\{\NAE_k\}^{LIT}, \nn)$},
    \item \textsc{Max E$k$-Lin} = \textsc{Max CSP$(\{\XOR\}^{NEG}, \nn)$},
    \item \textsc{Max $k$-Exact-SAT} = \textsc{Max CSP$(\{\EX_k\}^{LIT}, \nn)$}, where $\EX_k(x_1,\dots,x_k) = 1$ iff.~there is exactly one 1 in $(x_1,\dots,x_k)$.
\end{itemize}
Since $\deg(\glit) = \deg(\gneg) = \deg(\Gamma)$ it suffices to analyze the~characteristic polynomials for functions $\NAE_k$, $\XOR_k$, and $\EX_k$.
Let $e_i(x_1,\dots,x_k)$ denote $i$-th elementary symmetric polynomial, i.e., the~sum of~all degree-$i$ monomials on $k$ variables.

$\begin{array}{rll}
&&\\
\NAE_k(x_1,\dots,x_k) =& \sum_{i=1}^{k-1} (-1)^{i-1} e_i(x_1,\dots,x_k) & \text{for odd } k,\\ 
\NAE_k(x_1,\dots,x_k) =& \sum_{i=1}^{k} (-1)^{i-1} e_i(x_1,\dots,x_k) & \text{for even } k, \\
\XOR_k(x_1,\dots,x_k) =& \sum_{i=1}^{k} (-2)^{i-1} e_i(x_1,\dots,x_k) & \text{for all } k, \\
\EX_k(x_1,\dots,x_k) =& \sum_{i=1}^{k} i\cdot (-1)^{i-1} e_i(x_1,\dots,x_k) & \text{for all } k.\\&&
\end{array}$

\mic{\noindent It is easy to check these formulas \bmp{using the~binomial theorem. We} present the~argument for $\XOR_k$ as~an~example.
Suppose the~number of~1s in the~vector $(x_1,\dots,x_k)$ is $\ell$.
Then $e_i(x_1,\dots,x_k)$ equals $\binom \ell i$ for $i \le \ell$ and $0$ for $i > \ell$.
The formula~for $\XOR_k(x_1,\dots,x_k)$ becomes $\sum_{i=1}^{\ell} (-2)^{i-1} \binom \ell i = 
(-\frac{1}{2})\cdot \big(\sum_{i=0}^{\ell} (-2)^{i} \binom \ell i - 1 \big) = (-\frac{1}{2})\cdot((-1)^\ell - 1)$ which is 1 for odd $\ell$ and 0 for even $\ell$, as expected.}

By these identities we~deduce that the~{optimal compression exponent} for \textsc{Max $k$-NAE-SAT} is $k$ in the~even case, $k-1$ in the~odd case, and the~{optimal compression exponent} for both \textsc{Max E$k$-Lin} and \textsc{Max $k$-Exact-SAT} is $k$.

\mic{An example of~a~non-symmetric constraint with a~\bmp{non-trivial} upper bound on its degree is $f_k$, with $\ar(f_k) = 3^k$, defined recursively: $f_0(x) = x$ and $f_k(x_1, \dots x_{3^k}) = \NAE_3\big(f_{k-1}(x_1, \dots x_{3^{k-1}}),$ $f_{k-1}(x_{3^{k-1}+1},$ $\dots x_{2\cdot 3^{k-1}}), f_{k-1}(x_{2\cdot 3^{k-1}+1}, \dots x_{\bmp{3^{k}}})\big)$.
Since $\deg(\NAE_3) = 2$, we~have $\deg(f_k) = 2^k = \ar(f_k)^{\log_3(2)}$~\cite{NisanS94}.}

It is tempting to seek a~concise characterization of~$\Gamma$ for which one can obtain non-trivial bound on $\deg(\Gamma)$ and therefore a~non-trivial compression for \cspnn{}, where
by non-trivial we~mean a~bound of~the~form $\deg(f) \le \ar(f) - 1$ for functions depending on all the~coordinates.
Unfortunately, as far as we~are aware no such conditions are known, \bmp{not even} for symmetric polynomials induced by symmetric constraints.
There exist some interesting partial results though, e.g., that if the~number of~variables~$k$ is a~prime minus one then the~degree is always~$k$, and in general $\deg(f) \ge k - \Oh(k^{0.548})$~\cite{GathenR97}.
On the~other hand, there are infinitely many symmetric functions for which $\deg(f) = \ar(f) - 3$~\cite{GathenR97}.
When it comes to non-symmetric functions, there exist infinitely many examples with $\deg(f) \le \log(\ar(f))$~\cite{Simon83}, which is asymptotically the~lowest upper bound possible~\cite{NisanS94}.

\paragraph{Negative weights}
For the~sake of~completeness, we~show that an analogous classification holds for \cspzz{}, that is, whenever \cspzz{} is NP-hard, then the~upper bound from Theorem~\ref{thm:csp-ub} is essentially tight.
The dichotomy theorem for $\mathbb{W} = \zz$ can be stated in a~simpler manner, as the~problem becomes NP-hard whenever $\deg(\Gamma) \ge 2$~\cite{JonssonK07} and the~case $\deg(\Gamma) = 1$ reduces to linear function maximization.
This dichotomy follows also from the~reduction below. 

\begin{theorem}\label{thm:csp-zz-lb}
Let non-trivial $\Gamma$ be such that $d = \deg (\Gamma) \ge 2$.
Then there is a~constant $c$ such that \cspzzc{} does not admit a~compression of~size  $\Oh\left(n^{d-\eps}\right)$ for any $\eps > 0$, unless $\text{NP} \subseteq \text{co-NP} / \text{poly}$.
\end{theorem}
\begin{proof}
Since $f$ and $\neg f$ cannot be 0-valid (or 1-valid) at the~same time, the~language $\gneg$ is neither 0-valid nor 1-valid.
We can thus apply Corollary~\ref{cor:additive} to get $\cs(\Gamma_{d\text{-SAT}}, \zz) \le_{ADD} \cs(\gneg, \zz)$.
We compose it with $\cs(\gneg, \zz) \le_{ADD} \cs(\Gamma, \zz)$ (Lemma~\ref{lem:signed-reductions}, point (1)) to obtain that an $\Oh\left(n^{d-\eps}\right)$-size compression for \cspzzc{} with sufficiently large $c$ entails the~same for \textsc{Max CSP$(\Gamma_{d\text{-SAT}}, \zz, c - \Oh(1))$}, which implies $\text{NP} \subseteq \text{co-NP} / \text{poly}$~\cite{DellM14}.
\end{proof}

\section{Consequences for exponential-time algorithms}
\label{sec:exponential}

As mentioned before, our framework of~reductions can be used to preserve the~exponential running time as well.
Namely, if $CS(\Gamma_1, \mathbb{W}_1) \le_{ADD} CS(\Gamma_2, \mathbb{W}_2)$, then an algorithm for \textsc{Max CSP}$(\Gamma_2, \mathbb{W}_2, c)$ with running time $T(n)$ entails an algorithm for \textsc{Max CSP}$(\Gamma_2, \mathbb{W}_2, c - \Oh(1))$ with running time $T(n + \Oh(1))$.
All the~constructed transformation, except from $\cs(\glit, \nn) \le_{LIN} \cs(\Gamma, \nn)$ (Lemma~\ref{lem:implement-lit}), are additive and in particular they work as long as negative weights are allowed.
Alternatively, we~can take advantage of~other properties of~particular constraint languages to remove the~negative weights.

\begin{lemma}\label{lem:exp-cycle}
Let $d = \deg(\Gamma) \ge 2$. Then $\cs(\Gamma_{d\text{-SAT}}, \nn) \le_{ADD} \cs(\Gamma, \zz) \le_{ADD} \cs(\Gamma_{d\text{-SAT}}, \nn)$, that is, these constraint systems are equivalent with respect to relation $\le_{ADD}$.
\end{lemma}
\begin{proof}
Recall that $\gneg$ can be neither 0-valid nor 1-valid so it satisfies the~conditions for $\Gamma_2$ in Corollary~\ref{cor:additive}.
The same holds for $\Gamma_{d\text{-SAT}}$.
We also take advantage of~the~fact that this language can express negated literals, i.e., $\left (\Gamma_{d\text{-SAT}} \right)^{LIT} = \Gamma_{d\text{-SAT}}$.
We have the~following cycle of~reductions.
\begin{align*}
\cs(\Gamma_{d\text{-SAT}}, \zz) \quad &\le_{ADD} \quad \cs(\Gamma_{d\text{-SAT}}, \nn) \quad & \text{Lemma~\ref{lem:unsigned-lit} for $(\Gamma_{d\text{-SAT}})^{LIT} = \Gamma_{d\text{-SAT}}$}\\
&\le_{ADD} \quad \cs(\gneg, \zz) \quad & \text{Corollary~\ref{cor:additive}} \\
&\le_{ADD} \quad \cs(\Gamma, \zz) \quad & \text{Lemma~\ref{lem:signed-reductions}, point (1)} \\
&\le_{ADD} \quad \cs(\Gamma_{d\text{-SAT}}, \zz) \quad & \text{Corollary~\ref{cor:additive}} & & \qedhere
\end{align*}
\end{proof}

\begin{theorem}
For each $d \ge 2$ and any constant $\alpha~> 1$ either all the~following problems admit an~$\alpha^nn^{\Oh(1)}$ algorithm for all $c$, or none of~them do:
\begin{enumerate}
    \item \cspnnpar{\Gamma_{d\text{-SAT}}},
    \item \cspzzc{} for any $\Gamma$ with $\deg(\Gamma) = d$,
    \item \cspnnc{} for any $\Gamma$ with $\deg(\Gamma) = d$ such that $\gneg = \Gamma$ or $\glit = \Gamma$.
\end{enumerate}
\end{theorem}
\begin{proof}
As noted in Proposition~\ref{pro:exp}, if $\cs(\Gamma_1, \mathbb{W}_1) \le_{ADD} \cs(\Gamma_2, \mathbb{W}_2)$ and \textsc{Max CSP}$(\Gamma_2, \mathbb{W}_2, c)$ admits an algorithm with running time $\alpha^nn^{\Oh(1)}$, then the~same holds for \textsc{Max CSP}$(\Gamma_1, \mathbb{W}_1, c - \Oh(1))$.
The equivalency between (1) and (2) has been proven in Lemma~\ref{lem:exp-cycle}.
Since (2) is more general than (3), it suffices to reduce (2) into (3).
Lemma~\ref{lem:signed-reductions}, point (2), provides the~reduction $\cs(\Gamma, \zz) \le_{ADD} \cs(\gneg, \nn)$ for the~case $\gneg = \Gamma$ and Lemma~\ref{lem:unsigned-lit} provides the~reduction $\cs(\Gamma, \zz) \le_{ADD} \cs(\glit, \nn)$ for the~case $\glit = \Gamma$.
\end{proof}

First corollary of~this theorem is that problems of~form \cspzz{} are divided into equivalence classes with respect to the~optimal running time.
In particular, solving any \cspzz{} with $\deg(\Gamma) \ge 3$ in time $\Oh(2^{(1-\eps)n})$  for any $\eps > 0$ contradicts the~\textsc{Max 3-SAT Hypothesis}.
Also, the~hypothesis remains equivalent if we~replace \textsc{Max 3-CNF-SAT} with \textsc{Max 3-Lin SAT} or \textsc{Max 3-Exact SAT} with only positive weights, because their constraint languages satisfy $\gneg = \Gamma$ and $\glit = \Gamma$, respectively.
Another corollary is that
improving the~running time $\Oh\big(2^{\frac{\omega~n}{3}}\big)$ for \textsc{Max Cut} or \textsc{Max DiCut} with integer weights or \textsc{Max 3-NAE-SAT} with positive weights would imply an analogous breakthrough for \textsc{Max 2-CNF-SAT}.

\mic{Alman and Williams \cite{AlmanW20} have noted that it is not known how to improve the~running time for \textsc{Max 3-CNF-SAT}, even for instances with a~linear number of~clauses. They have therefore formulated a~stronger hypothesis.
The \textsc{Sparse Max 3-SAT Hypothesis} states that there exists $c > 0$ such that \textsc{Max 3-CNF-SAT} with $cn$ clauses does not admit an~$\Oh(2^{(1-\eps)n})$-time algorithm for any $\eps > 0$.
Similarly, their \textsc{Sparse Max 2-SAT Hypothesis} states that one cannot beat running time $\Oh\big(2^{\frac{\omega~n}{3}}\big)$ for \textsc{Max 2-CNF-SAT} with $cn$ clauses.
Observe that our reductions preserve the~property of~having $\Oh(n)$ different constraint applications (condition (2) in Definition~\ref{def:transform}).
Therefore as long as we~allow negative weights at constraint applications, we~can replace \textsc{Max 2-CNF-SAT} (resp.~\textsc{Max 3-CNF-SAT}) in this hypothesis with \cspzz{} for any constraint language~$\Gamma$ of~degree 2 (resp.~3) to obtain an equivalent statement.}


\section{Conclusions and open problems}
We have provided a~complete characterization of~the~optimal compression for \textsc{Max CSP}($\Gamma$) in the~case of~a~Boolean domain.
a~natural question arises about larger domains.
Our approach does not transfer even to the~case with a~domain of~size 3, since there is no unique way to represent functions $\{0,1,2\}^k \rightarrow \{0,1\}$ as polynomials.
One may consider, e.g., embedding to a~Boolean domain or using non-multilinear polynomials, but it is not clear which approach leads to the~optimal degree and how to find accompanying lower bounds.

On the~exponential-time front, we~have showed that \textsc{Max $d$-CNF-SAT} is as hard as any \textsc{Max CSP} of~degree $d$ as long as negative weights are allowed.
\mic{Although we~were able to get rid of~the~latter assumption in several cases, there is still a~gap in this classification: does improving the~running time for any degree-$d$ \cspnn{} imply an improvement for \textsc{Max $d$-CNF-SAT}?}

\bibliographystyle{abbrvurl}
\bibliography{max-csp}

\end{document}